\documentclass[11pt,letterpaper]{article}

\usepackage[linesnumbered,ruled,vlined]{algorithm2e}
\usepackage{amsmath,amssymb,mathrsfs,amsthm}
\usepackage{microtype}
\usepackage{booktabs}
\usepackage{graphicx}
\usepackage{enumitem}
\usepackage{bbm}
\usepackage{blkarray}
\usepackage{bm}
\usepackage{listings}
\usepackage{caption}
\usepackage{subcaption}
\usepackage{econometrics}
\usepackage{setspace}
\usepackage{multirow}
\usepackage{xr}
\usepackage[authoryear]{natbib}
\usepackage[colorlinks=true,linkcolor=blue,citecolor=blue]{hyperref}
\usepackage[toc,title,titletoc,header]{appendix}
\usepackage[left=1in,right=1in,top=0.8in,bottom=1in]{geometry}
\usepackage{lineno}
\setpagewiselinenumbers

\theoremstyle{definition}
\newtheorem{example}{Example}[section]
 
\newtheorem{theorem}{Theorem}[section]

\newtheorem{proposition}{Proposition}[section]
\newtheorem{remark}{Remark}[section]
\newtheorem{assumption}{Assumption}

\newcommand{\R}{\mathbb{R}}
\newcommand{\N}{\mathbb{N}}

\newcommand{\bA}{\bm{A}}
\newcommand{\bD}{\bm{D}}
\newcommand{\bL}{\bm{L}}
\newcommand{\bw}{\bm{w}}

\newcommand{\bZ}{\bm{Z}}
\newcommand{\bY}{\bm{Y}}
\newcommand{\bone}{\bm{1}}
\newcommand{\bzero}{\bm{0}}
\newcommand{\Var}{\mathrm{Var}}
\newcommand{\TTE}{\mathrm{TTE}}
\newcommand{\M}{\mathrm{M}}

\newcommand{\WRI}{\mathrm{WRI}}

\newcommand{\CR}{\mathrm{CR}}
\newcommand{\MPR}{\mathrm{MPR}}
\newcommand{\dmax}{d_\text{max}}
\newcommand{\dmin}{d_\text{min}}

\newcommand{\ie}{{i.e.}}

\newcommand{\Ber}{\mathrm{Bernoulli}}
\newcommand{\iidsim}{\overset{iid}{\sim}}
\newcommand{\htau}{\hat{\tau}}

\newcommand{\outnode}{\mathcal{O}}
\newcommand{\innode}{\mathcal{I}}
\newcommand{\toutnode}{\tilde{\mathcal{O}}}
\newcommand{\tinnode}{\tilde{\mathcal{I}}}

\SetKwInput{KwInput}{Input}                
\SetKwInput{KwOutput}{Output}

\title{\bf Adaptive Weighted Random Isolation (AWRI): a simple design to estimate causal effects under network interference}

\usepackage{authblk}

\author[a]{Changhao Shi\textsuperscript{*}}
\author[b]{Haoyu Yang\textsuperscript{*}}
\author[c]{Yichen Qin}
\author[a]{Yang Li\textsuperscript{\dag}}
\affil[a]{Center for Applied Statistics and School of Statistics, Renmin University of China}
\affil[b]{Biostatistics, Harvard T.H. Chan School of Public Health, Harvard University}
\affil[c]{Department of Operations, Business Analytics, and Information Systems, University of Cincinnati}

\footnotetext[1]{These two authors contributed equally.}
\footnotetext[2]{Corresponding author: \texttt{yang.li@ruc.edu.cn}.}
\date{\today}

\begin{document}
\maketitle

\onehalfspacing

\vspace{-1cm}
\begin{abstract}

    Recently, causal inference under interference has gained increasing attention in the literature. In this paper, we focus on randomized designs for estimating the total treatment effect (TTE), defined as the average difference in potential outcomes between fully treated and fully controlled groups. We propose a simple design called weighted random isolation (WRI) along with a restricted difference-in-means estimator (RDIM) for TTE estimation.  Additionally, we derive a novel mean squared error surrogate for the RDIM estimator, supported by a network-adaptive weight selection algorithm. This can help us determine a fair weight for the WRI design, thereby effectively reducing the bias. Our method accommodates directed networks, extending previous frameworks. Extensive simulations demonstrate that the proposed method outperforms nine established methods across a wide range of scenarios.

\end{abstract} 


\newpage
\onehalfspacing

\section{Introduction}

Causality is a fundamental concept in statistical research. Over the past few decades, a wide range of methods have been developed to estimate causal effects across a variety of settings. For a review, see \cite{pearl2009causality, ImbensRubin2015,ding2018causal, ding2024first}. Most of these methods are grounded in the assumption of no interference, commonly referred to as the stable unit treatment value assumption (SUTVA), which posits that the treatment assignment of one unit does not affect the outcomes of others \citep{cox1958planning, Rubin1980, ImbensRubin2015}. However, this assumption is often unrealistic in many real-world scenarios, such as the evaluation of  public health interventions \citep{vanderweele2012mapping, alexandria2021}, the study of economic policies \citep{cai2015social, leung2020}, the public policy interventions \citep{paluck2016changing, egami2021spillover, grossi2020synthetic} and online A/B tests \citep{basse2018, saint2019using, liu2022adaptive}, etc, where units interact through social networks. To address these complexities, numerous new methods have emerged in recent years \citep{HudgensHalloran2008, AronowSamii2017, BasseFeller2018, karwa2018systematic, hu2022average, Leung2022, gao2023causal, savje2024causal}.

In this paper, we focus on randomized experiment design, often considered the ``gold-standard'' of causal inference \citep{fisher1935}. Among the various types of causal effects under interference, our primary concern is the total treatment effect (TTE), defined as the average difference in potential outcomes between fully treated and fully controlled groups. It serves as a measure of the overall impact of a new policy or intervene \citep{UganderKarrerBackstrom2013, leung2022a, yu2022estimating, cortez2022staggered, ugander2023randomized, viviano2023causal, cai2023independent, cortez2024combining}. 

There are two main lines of methods to estimate causal effects through randomized design: one based on inverse probability weighted estimators (IPW) \citep{HorvitzThompson1952, Hajek1971} and the other based on difference-in-means estimators (DIM). The former inherits the naturally low bias of classical IPW estimators but this comes at the expense of high variance due to units' low exposure probabilities when interference is strong. The latter retains the simplicity and low variance of classical estimators but suffers from high bias in the presence of strong interference. Consequently, there are two totally different directions to design randomization mechanisms: the first aims to reduce variance caused by low exposure probabilities \citep{UganderKarrerBackstrom2013, BasseFeller2018, Leung2022, liu2022adaptive, yu2022estimating, ugander2023randomized, gao2023causal}, while the second focuses on minimizing bias introduced by interference \citep{jagadeesan2020, leung2022a, viviano2023causal, cai2023independent}. 

Each of these methods has its own advantages and drawbacks, and currently, no single method can confidently claim superiority over others. Although some new methods, such as cluster-based randomized designs, seem more reasonable than the original ones, their theoretical guarantees remain incomplete due to the complexity of dependencies. Only in the simplest Bernoulli design, where treatments are assigned by repeatedly tossing an uneven coin, can IPW-type estimators be rigorously analyzed, as shown in \cite{Leung2022} and \cite{gao2023causal}, though under quite stringent assumptions. Cluster-based randomized designs aim to reduce interference between clusters by properly partitioning networks into clusters. By treating each cluster as a new ``unit'' and assigning treatments at the cluster level, one can expect a reduction in bias or variance, depending on the choice of estimators. The graph cluster randomization (GCR) design proposed by \cite{UganderKarrerBackstrom2013}, the randomized graph cluster randomization (RGCR) proposed by \cite{ugander2023randomized}, the causal clustering (CC) design proposed by \cite{viviano2023causal}, and the independent-set (IS) design proposed by \cite{cai2023independent} are four typical examples.

The GCR uses ``3-net clustering'' to reduce variance by improving the exposure probabilities of high-degree units, while the RGCR adds extra randomness during the clustering stage, allowing complete randomization (CR) in the subsequent stage, further reducing variance. However, both designs require exposure probabilities to be estimated through simulations, which is computationally expensive. The CC design seeks to find the optimal clustering by minimizing a novel objective function that balances bias and variance, but it requires extensive prior information to tune parameters, making it challenging to implement in practice. The IS design partitions the network into an independent set, where units are used to construct estimators, and an auxiliary set, where units provide specific interference to those in the independent set, effectively reducing the bias. This design is simple and effective in many settings, but it is limited by its strong dependence on the correct specification of the interference model, which is often difficult to satisfy in practice.

All the methods mentioned above are designed for undirected networks, which is generally not realistic. Additionally, there seems to be a trend towards increasingly complex designs, which makes interpretation and implementation more difficult. In this paper, to address these issues, we introduce a simple design with a straightforward DIM estimator that accommodates directed networks and outperforms other methods across a wide range of scenarios.

To introduce our methods, we start with a simple scenario where units have been divided into clusters, and those from different clusters are assumed not to interfere with each other. This assumption is known as partial interference \citep{HudgensHalloran2008}, and many effective results have been achieved under this framework \citep{BasseFeller2018}. Inspired by this situation, we aim to create a similar environment for some units. These units must satisfy two conditions: (i) they can be assigned any exposure (any treatment assignments from themselves and their neighbors) without influencing each other; (ii) they are representative of the entire population. Combining these two properties allows us to use a simple difference-in-means estimator restricted to these units to estimate the TTE, with the expectation that the convergence rate of its mean square error (MSE) will be comparable to that in the partial interference scenario.

For the first condition, we need to ``isolate'' units, ensuring that the isolated units are sufficiently distant from each other. To address this, we extend the RGCR design of \cite{ugander2023randomized} to accommodate directed networks, calling this design weighted random isolation (WRI). The second condition is more challenging. Technically, we want to obtain a random sample that is representative of the population. We begin with a simple fact: a simple random sample is always representative. However, it is not feasible here because it contradicts the first condition (which manages the complex dependencies caused by interference and thus cannot be compromised). Therefore, we aim to generate a random sample that is ``similar'' to a simple random sample. For instance, this sample should be able to balance network-based covariates, such as the degrees of units. We achieve this by adjusting the weight of the WRI, which controls the sample distribution. By incorporating network-based information into the potential outcomes, we derive a novel MSE surrogate. By minimizing this surrogate, we can obtain the optimal weight for the WRI design. Finally, inspired by a simple example, we propose some network-based weights as candidate weights which perform well in practice. Since the weight selection process is network-adaptive, we refer to the entire design as adaptive weighted random isolation (AWRI).

The remainder of this paper is organized as follows: Section \ref{sec: notation} outlines the basic notations and assumptions in network experiments. Section \ref{sec: random isolation} proposes the random isolation (RI) design and restricted estimators. In section \ref{sec: representative isolated set}, we introduce the weighted random isolation (WRI) design and present the adaptive weight selection algorithm based on minimizing a novel MSE surrogate. Section \ref{sec: simulation} shows the numerical results. Section \ref{sec: conclusion} offers concluding remarks and suggests directions for future research. All proofs and additional results are included in the Appendix.

\section{Setup}
\label{sec: notation}

Throughout this work we consider a finite population setting, where all potential outcomes are unknown but fixed and the treatment assignment is the only source of randomness, also known as the design-based framework \citep{ImbensRubin2015, AronowSamii2017, AbadieAtheyImbens2020, Leung2022, gao2023causal}. In the presence of interference, units interact each other through a directed, unweighted network $G=(V, E)$, where $V=[n]=\{1,\dots,n\}$ denotes the set of units and $E$ denotes the set of edges that represent the underlying interference, i.e. $(i,j)\in E$ if unit $i$ interfere with unit $j$. Equivalently, we describe it by an adjacency matrix $\bA=(A_{ij})_{n\times n}$ with the $(i, j)$th entry $A_{ij} \in \{0, 1\}$ indicating the connection from unit $i$ to $j$, i.e. $A_{ij}=1$ if $(i,j)\in E$. Let $\innode_i$ denote the set of units that interfere with unit $i$ and $\outnode_i$ denote the set of units that unit $i$ interferes with, i.e., $\innode_i=\{j,(j,i)\in E\}$ and $\outnode_i=\{j,(i,j)\in E\}$. Let $\tilde{\innode}_i$ and $\tilde{\outnode}_i$ denote $\innode_i\cup \{i\}$ and $\outnode_i\cup \{i\}$, respectively. We use $d_i$ denote the in-degree of unit $i$, i.e., $d_i = |\innode_i|$. In an undirected network, we simply term it as the degree. We assume the network is fixed and known to the researcher.

The treatment assignment is denoted by $\bZ = (Z_i)_{i=1}^n \in \{0,1\}^n$, where each $Z_i$ indicates whether unit $i$ is assigned to the treatment. The potential outcome of unit $i$, denoted as $Y_i(\bZ)$, represents the outcome when the treatment assignments for all $n$ units are given by $\bZ$. This implies that $Y_i(\bZ)$ depends not only on $Z_i$, but also on the treatment assignments of all other units, capturing the essence of ``interference''.

Throughout this paper, we focus on the total treatment rffect (TTE), which is defined as the average difference in potential outcomes between fully treated and fully controlled groups. Specifically, the TTE is defined as
\begin{equation}
    \tau = \frac{1}{n}\sum_{i=1}^n\tau_i = \frac{1}{n}\sum_{i=1}^n\left(Y_i(\bone)-Y_i(\bzero)\right),
\end{equation}
where $\bone= \bone_n =(1)_{i=1}^n$ and $\bzero=\bzero_n = (0)_{i=1}^n$ are vectors of length $n$. The TTE cannot be directly estimated because $\{Y_i(\bone),i\in [n]\}$ and $\{Y_i(\bzero),i\in [n]\}$ can not be observed simultaneously. To address this issue, we introduce the Full Neighborhood Interference (FNI) Assumption, which is a popular choice in the literature \citep{UganderKarrerBackstrom2013,ugander2023randomized,yu2022estimating,viviano2023causal,cai2023independent}.

\begin{assumption}[Full Neighborhood Interference (FNI)]
    \label{asu: fni}
For all $i\in [n]$, $\bZ,{\bZ}^\prime\in\{1,0\}^n$, if $Z_j=Z^{\prime}_j$ for each $j \in \tinnode_i$, then $Y_i(\bZ)=Y_i({\bZ}^{\prime})$.
\end{assumption}
 
Assumption \ref{asu: fni} states that the potential outcomes of a unit depend only on the treatment assignments of its neighborhood and itself. We use ${\bZ}_{\tinnode_i}=(Z_i,i\in \tinnode_i)$ indicate a sub-vector of $\bZ$. We call every possible value of ${\bZ}_{\tinnode_i}\in \{0,1\}^{|\tinnode_i|}$ is an ``exposure'' of unit $i$ \citep{AronowSamii2017}. Under this assumption, one can observe $Y_i(\bone)$ if ${\bZ}_{\tinnode_i}=\bone_{|\tinnode_i|}$ and observe $Y_i(\bzero)$ if ${\bZ}_{\tinnode_i}=\bzero_{|\tinnode_i|}$, which makes TTE can be estimated under specific designs.

\section{Random isolation and restricted estimators}
\label{sec: random isolation}

To motivate our method, we begin with a simple scenario where $n$ units are divided into $K$ clusters, with interference occurring only within clusters, referred to as partial interference \citep{HudgensHalloran2008}. In this situation, one can implement the two-stage randomization and then estiamte causal effects with standard inverse probability weighted (IPW) estimators \citep{BasseFeller2018}. As a result, the mean squared error (MSE) of the estiamtors is of the same order as $O(K^{-1})$ instead of the usual $O(n^{-1})$ (Proposition 4.1 and Proposition 5.1 in \cite{BasseFeller2018}). This suggests that, in the presence of interference, the number of independent clusters is a key quantity regarding the large sample properties. We should consider it as the effective sample size, which is much smaller than the actual sample size $n$. 

However, under more complex interference, natural clusters are typically absent, which motivates us to artificially partition the network to generate clusters that function similarly to those under partial interference. The artificial clusters should have the following merits: each cluster should contain at least one unit $i$ such that $\tinnode_i$ is included in the cluster. This means that under Assumption \ref{asu: fni}, all possible exposures of unit $i$, ${\bZ}_{\tinnode_i}\in \{0,1\}^{|\tinnode_i|}$, can be assigned solely within this cluster, independent of the treatment assignments of units in other clusters. This aligns with the spirit of clusters in partial interference. A toy example is shown in Figure \ref{fig: toyexampleRI}.

\begin{figure}[!htb]
    \centering
    \includegraphics[width=0.95\textwidth]{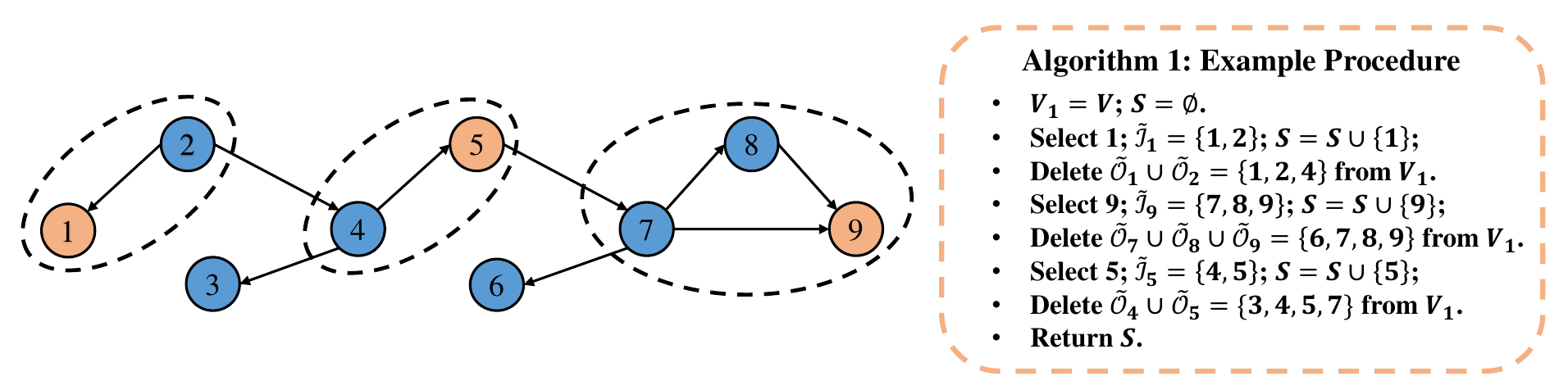}
    \caption{A toy example to illustrate the random isolation method. The left part shows the spirit of clusters under general interference, where dotted ovals indicate clusters and orange units indicate isolated units. The right part is an example procedure of Algorithm \ref{alg: random isolation}.} 
    \label{fig: toyexampleRI}
\end{figure}

To achieve this, we initialize $S=\varnothing$ and $V_1 = V$. Select one unit $i$ from $V_1$ uniformly at random and add it to $S$. Then, remove $\tinnode_i$ from $V_1$. This process is repeated until no unit is left in $V_1$. An example procedure is shown in Figure \ref{fig: toyexampleRI}. In this way, units in $S$ are sufficiently ``isolated'' such that their exposures can be assigned independently. We call this procedure random isolation (RI), outlined in Algorithm \ref{alg: random isolation}. The output of RI, $S$, is called isolated set and elements of $S$ are called ``isolated units''. This procedure is an extension of the ``uniform-3-net clustering'' in \cite{ugander2023randomized} and ``independent-set design'' in \cite{cai2023independent}.

\begin{algorithm}[!htb]
    \caption{Random Isolation (RI)}
    \label{alg: random isolation}
    \KwInput{network $G=(V, E)$}
    \KwOutput{isolated set $S$}
    $S\leftarrow \varnothing$\\
    $V_1 \leftarrow V$\\
    \While{$|V_1|>0$}
    {
        select $i$ uniformly at rondom from $V_1$ \label{line: uniform}\\
        $S\leftarrow S\cup\{i\}$\\
        $V_1\leftarrow V_1\setminus \cup_{l\in \tinnode_i}\toutnode_l$
    }
    \Return{$S$}
\end{algorithm}

Under Assumption \ref{asu: fni}, the advantage of RI design is apparent: the isolated units do not interfere each other at all and if we only care about the TTE of sub-population $S$, then we can estimate them as if there is no interference. Specifically, define the TTE of $S$ as
\begin{equation}
\label{eq: TTE of S}
\tau_S = \frac{1}{|S|}\sum_{i\in S}\tau_i = \frac{1}{|S|}\sum_{i=1}^n\tau_i I(i\in S) = \frac{1}{|S|}\sum_{i=1}^n \left(Y_i(\bone)-Y_i(\bzero)\right)I(i\in S).
\end{equation}
Then we can proceed complete randomization at the cluster level in 2 steps: (i) implement complete randomization (CR) on the isolated set $S$ with treatment group size equal to $\lfloor|S|/2\rfloor$ to get the treatment group $S_1$, written as $S_1 = \CR(S,\lfloor|S|/2\rfloor)$; (ii) for each unit in $S_1$, assign treatment to all units in $S_1$ and its 1-neighborhood. The procedure is outlined in Algorithm \ref{alg: cr}.

\begin{algorithm}[!htb]
    \DontPrintSemicolon
    \caption{Complete randomization at the cluster level}
    \label{alg: cr}
    \KwInput{isolated set $S\subset V$}
    \KwOutput{treatment assignment $\bZ\in\{1,0\}^n$}
    $\bZ \leftarrow \bzero_n$ \\
    $S_1 \leftarrow \CR(S,\lfloor|S|/2\rfloor)$ \label{line: cr} \\
    \For{$i \in S_1$}{ \label{line: cr2}
        ${\bZ}_{\tinnode_i} \leftarrow \bone_{|\tinnode_i|}$ \label{line: cr3}
    }
    \Return{$\bZ$}
\end{algorithm}

Then we get observations $\{Z_i,Y_i,i\in S\}$, where $Y_i = Y_i(\bone)Z_i + Y_i(\bzero)(1-Z_i)$. And we can borrow the traditional difference-in-means estimator $\hat{\tau}$ to estimate $\tau_S$, where
\begin{equation}
\label{eq: dim estimator}
\hat{\tau} = \frac{1}{|S_1|}\sum_{i\in S_1}Y_i Z_i - \frac{1}{|S_0|}\sum_{i\in S_0}Y_i (1-Z_i).
\end{equation}
The classical properties of $\hat{\tau}$ are summarized in Theorem \ref{theo: classic resutls of diff} in Appendix \ref{app: additional results}. This theorem demonstrates that we can reliably estimate $\tau_S$ for any sub-population $S$ selected by random isolation (RI) method. Additionally, one can replace complete randomization with matched-pairs randomization (MPR) in Line \ref{line: cr} of Algorithm \ref{alg: cr}, and substitute the difference-in-means estimator $\hat{\tau}$ with the matched-pairs estimator $\hat{\tau}_m$ to imrpove the performance of finite population. Details are discussed in Appendix \ref{app: matched}.

However, our primary interest lies in estimating the TTE of the original population, $\tau$, which may differ significantly from $\tau_S$. The intuition is that if $S$ is a simple random sample of $[n]$ with size $K$, then $\{\tau_i, i \in S\}$ is representative of $\{\tau_i, i \in [n]\}$. Thus, under the assumption of bounded potential outcomes, $\MSE(\tau_S)$ can be controlled by $O(K^{-1})$, which is the same order as $\E_S(\Var(\hat{\tau}|S))$ \citep{sarndal2003model}. This implies $\MSE(\htau) = O(K^{-1})$, comparable to the rate under partial interference. Details are presented in Theorem \ref{theo: simple random sample} in Appendix \ref{app: simple random sample}. This observation motivates us to select a representative sub-population $S$ that behaves as much as possible like a simple random sample. The next section will address this issue. To conclude this section, we provide some convenient notations.

In the rest of this paper, we use $\CR(S, n_1)$ indicate a complete randomization which maps a set $S$ and a number $n_1$ to a treatment group $S_1\subset S$ with size $n_1$. $\MPR(S, n_1)$ represents matched-pairs randomization, defined analogously. We name their corresponding estimators as restricted difference-in-means estimator (rdim) and restricted matched estimator (rmat) because both of them only use data restricted on the isolated set $S$.

\section{Choosing representative isolated sets}
\label{sec: representative isolated set}

\subsection{Weighted Random Isolation}

As mentioned in the last section, ideally, we hope the sub-population $S$ is ``like'' a simple random sample, which well represents the whole population. However, in the presence of complex interference, this cannot be achieved: some units have low probabilities of being sampled into $S$ due to their high degrees, and some units cannot be sampled into $S$ simultaneously because of their close proximity. 

To address this problem, we develop a novel sampling technique, adaptive weighted random isolation (AWRI), consisting of two parts: weighted random isolation (WRI, Algorithm \ref{alg: weighted random isolation}) and adaptive weight selcetion (Algorithm \ref{alg: adaptive weight selection}). We introduce WRI first. The basic idea is to assign each unit a specific probability of being sampled into $S$ in each round of samplings, i.e. substituting the uniform random mechanism with a more controlled random mechanism in line \ref{line: uniform} of Algorithm \ref{alg: random isolation}. By this way, researchers can adjust the distribution of $S$ to make it more representative.

However, although we can control the sampling probability in each round, we cannot precisely manage it over the entire process. The final sampling probability is determined by a complex stochastic process, which is hard to analyze. Nevertheless, we can still influence the overall ``trend'': if one unit has a higher probability of being sampled than another in each round, it will also have a higher probability of being sampled throughout the entire process. This can be easily achieved by a roulette wheel selection procedure \citep{LIPOWSKI20122193}. First, for $j\in[n]$, assign a weight $w_j$ to unit $j$. Set $V_1=V$. Then in each round, for each $j\in V_1$, sample unit $j$ with probability $w_j/\sum_{l\in V_1}w_l$. After sampling a unit $i$, remove $\cup_{l\in \tinnode_i}\toutnode_l$ from $V_1$. Repeat this process until $V_1$ is empty. The whole procedure is outlined in Algorithm \ref{alg: weighted random isolation}, where an equivalent version is provided, utilizing the properties of beta distribution (Proposition \ref{prop: beta}). This procedure is an extension of the ``weighted-3-net clustering'' in \cite{ugander2023randomized}. From now on, we use $\WRI(G, \bw)$ as a random function: $\WRI(G, \bw)$ maps the input network $G$ and weight $\bw$ into a random isolated set $S^{\bw}$.

\begin{proposition}
    \label{prop: beta}
    For independent random variables $X_i \sim \mathrm{Beta}(w_i, 1)$, $X_j \sim \mathrm{Beta}(w_j, 1)$, we have $\mathrm{P}(X_i > X_j) = {w_i}/{(w_i + w_j)}$ and $\max\{X_i, X_j\} \sim \mathrm{Beta}(w_i + w_j, 1)$.
\end{proposition}

\begin{algorithm}[!htb]
    \caption{Weighted Random Isolation (WRI)}\label{alg: weighted random isolation}
    \KwInput{network $G=(V, E)$, weight $\bm w \in \mathbb{R}^n_{\scriptscriptstyle \geq 0}$}
    \KwOutput{isolated set $S$}
    \For{$i \in V $ }{
    $X_i \overset{iid}{\sim} \mathrm{Beta}(w_i,1)$\\
    }
    $S\leftarrow \varnothing$\\
    $V_1\leftarrow V$\\
    \While{$|V_1|>0$}
    {
        $i\leftarrow\arg\max\{X_i, i\in V_1 \}$ \\
        $S\leftarrow S\cup\{i\}$\\
        $V_1\leftarrow V_1\setminus \cup_{l\in \tinnode_i}\toutnode_l$
    }
    \Return{$S$}
\end{algorithm}

\subsection{Weight selection based on a novel MSE surrogate}

Thus far, we've reduced the original problem of assigning each unit a proper probability to a seemingly simpler one: selecting an appropriate weight vector $\bw$ for WRI design. This is a typical optimization problem. To address it, we propose a novel mean squared error (MSE) surrogate function. By minimizing this function, we can identify an effective weight for the WRI design. To motivate our approach, we first introduce some key assumptions.

\begin{assumption}[Potential Outcomes Decomposition] \label{asu: potential outcomes decomposition} 
There exist functions $f_1$ and $f_0$ (not necessarily unique) such that, for all $i \in [n]$ and $n \in \N_{+}$, $Y_i(\bone) = f_1(d_i) + \varepsilon_{1i}$ and $Y_i(\bzero) = f_0(d_i) + \varepsilon_{0i}$, where $d_i$ is the in-degree of the $i$th unit and $\varepsilon_{1i}, \varepsilon_{0i} \in \R$.
\end{assumption}

Technically, Assumption \ref{asu: potential outcomes decomposition} just states a formal decomposition of potential outcomes and does not impose any real restrictions on $Y_i(\bone)$ and $Y_i(\bzero)$. For example, to make it hold, one can simply set $f_1=f_0\equiv 0$ and $\varepsilon_{1i}=Y_i(\bone)$, $\varepsilon_{0i}=Y_i(\bzero)$. However, this decomposition is nontrivial under many realistic settings, where potential outcomes exhibit specific patterns related to in-degrees such as \cite{ugander2023randomized} and \cite{parker2017optimal} assume potential outcomes are linear with respect to the absolute number of treated neighbors.

\begin{assumption}[Bounded Potential Outcomes]\label{asu: bounded potential outcomes}
    There exist positive constants $c_1$ and $c_2$ such that
    \begin{itemize} 
        \item[(i)] for all $d \in \N$, $|f_l(d)|\leq c_1$, $l=0, 1$;
        \item[(ii)] for all $i\in [n]$, $n\in \N_{+}$ and every $S \subseteq [n]$, $|\frac{1}{|S|}\sum_{i\in S}\varepsilon_{li}|\leq \frac{1}{\sqrt{|S|}}c_2$, $l=0, 1$.  
    \end{itemize}
    \end{assumption}

\begin{assumption}[Bounded Potential Outcomes*]
\label{asu: normal bounded outcomes}
There exists positive constant $c$ such that
$\left|Y_i({z})\right|\leq c<\infty$, for all $ n \in \N_{+}, i \in [n], {z} \in\{0,1\}^n$.
\end{assumption}

Assumption \ref{asu: bounded potential outcomes} is a slightly stronger version of the normal boundness assumption stated in Assumption \ref{asu: normal bounded outcomes} (Assumption 3 in \cite{Leung2022}, \cite{gao2023causal} and \cite{viviano2023causal}, etc). Here we bound ``pattern'' terms $\{f_l(d),l\in\{0,1\},d\in \N\}$ and ``noise'' terms $\{\varepsilon_{li},l\in{0,1}, i\in [n],n\ n\in \N\}$ separately.

\begin{remark}\label{rem: boundness}
    Assumption \ref{asu: bounded potential outcomes} is actually a finite population version of normality assumption of ``noise'' terms $\{\varepsilon_{li}:l\in{0,1}, i\in [n],n\ n\in \N\}$. Suppose $\varepsilon_{li}|S \overset{iid}{\sim} {N}(0,\sigma^2)$ for all $i\in[n], n\in\N_+, l=0,1$ and for each $S\subseteq [n]$. Then given $S$, $\frac{1}{|S|}\sum_{i\in S}\varepsilon_{li} \sim {N}(0,\sigma^2/|S|)$, which means $\frac{1}{\sqrt{|S|}}\sum_{i\in S}\varepsilon_{li} = O_p(1)$. Here, analogously, under finite population situation, we suppose $\frac{1}{\sqrt{|S|}}\sum_{i\in S}\varepsilon_{li} = O(1)$. 
\end{remark}

Let $\dmax$ indicates the maximum of in-degrees of network $G$. Let $P_{S_1^{\bw}}$, $P_{S_0^{\bw}}$ and $P_G$ indicate probabilistic mass functions (PMF) of $\{d_i, i\in S_1^{\bw}\}$, $\{d_i,i\in S_0^{\bw}\}$ and $\{d_i, i \in [n]\}$, respectively. $P_{S_1^{\bw}}$ and $P_{S_0^{\bw}}$ are random because of the randomness of $S_1^{\bw}$ and $S_0^{\bw}$. The superscript $\bw$ indicates they depend on the weight $\bw$ which is used to implement WRI algorithm. We use $\htau_{\bw}$ indicate the restricted difference-in-means estimator under $\text{WRI}(G, \bw)$ design.

Now, with these preparations, we can give $\MSE(\hat{\tau}_{\bw})$ an intuitive upper bound which can help us understand the key point of this problem and can motivate an heuristic algorithm to choose optimal weight.

\begin{theorem}\label{theo: mse}
    Suppose that Assumptions \ref{asu: fni}, \ref{asu: potential outcomes decomposition}, \ref{asu: bounded potential outcomes} hold. Then
    \begin{normalsize}
        \begin{equation*}
            \MSE(\hat{\tau}_{\bw}) \leq 2c_1^2(d_\text{max}+2)^2\left(\E||P_{S_1^{\bw}} - P_G||_2^2 + \E||P_{S_0^{\bw}} - P_G||_2^2\right) + 16c_2^2\left(\E{|S_1^{\bw}|}^{-1}+\E{|S_0^{\bw}|}^{-1}\right).
        \end{equation*}
    \end{normalsize}
\end{theorem}
\begin{proof}
    See Appendix \ref{app: proof of theorem 4.1}.
\end{proof}

Theorem \ref{theo: mse} formulates an intuitive result: typically, if the ``distribution'' of sub-population is closer to the real ``distribution'' and if sample size of sub-population becomes larger, then we can expect estimators perform better. Here, $\MSE(\hat{\tau}_{\bw})$ is controlled by two parts: the first one is related to the ``pattern'' terms $\{f_l(d),l\in\{0,1\},d\in \N\}$ while the second one is related to the ``noise'' terms $\{\varepsilon_{li},l\in{0,1}, i\in [n],n\ n\in \N\}$, and coefficients $2c_1^2(d_\text{max}+2)^2$ and $16c_2^2$ reflect their corresponding importance. Specifically, if potential outcomes $\{Y_i(1),Y_i(0)\}$ strongly depend on the in-degree $d_i$ for all $i\in [n]$, then $c_2^2$ may become negligible compared to $c_1^2(d_\text{max}+2)^2$. On the contrary, if there is no any patterns between potential outcomes and in-degrees, $c_1$ could be rather small such that the second part becomes dominated. In most cases, we have no knowledge about the real situation, i.e., $c_1$ and $c_2$ are both unknown, thus it is reasonable to consider these two parts equally. We will revisit two classic examples to interpret this theorem more clearly.

\begin{example}[SUTVA]
    \label{exam: mse-sutva}
    Assume there are $n$ units. Under SUTVA, all units have in-degree $d_i = 0$, so the estimation of in-degree distribution is perfect, which means $||P_{S_1^{\bw}} - P_G||_2^2 = ||P_{S_0^{\bw}} - P_G||_2^2 = 0$. And $|S_1^{\bw}|=|S_0^{\bw}| = n/2$ (suppose $n$ is even for simplicity). Therefore, by Theorem \ref{theo: mse}, $\MSE(\hat{\tau}_{\bw}) \leq O(n^{-1})$, corresponding to the classic theory (Example 3.1 in \cite{ding2024first}).
\end{example}

\begin{example}[Partial Interference]\label{exam: partial interference}
Assume there are $K$ clusters, each of which has $n_c$ units. We use $ki$ to indicate the $i$th unit in the $k$th cluster. Under Partial Interference Assumption, the interference only exists in the same cluster, and does not exist between different clusters. For simplicity, assume all units in the same cluster are fully connected so their in-degrees $d_{ki}\equiv n_c-1$ for all $k \in [K]$ and $i \in [n_c]$. It's easy to see that the estimation of in-degree distribution is perfect as well like under SUTVA. And $|S_1^{\bw}|=|S_0^{\bw}| = K/2$ (suppose $K$ is even for simplicity). Therefore, by Theorem \ref{theo: mse}, $\MSE(\hat{\tau}_{\bw}) \leq O(K^{-1})$, corresponding to the well-known results (Proposition 5.1 in \cite{BasseFeller2018} and the proof of Theorem 4 in \cite{HudgensHalloran2008}).
\end{example}

\begin{remark}\label{rem: thm4.1}
    When we have more covariates, Assumption \ref{asu: potential outcomes decomposition} and \ref{asu: bounded potential outcomes} can be easily extended. For example, one can write $Y_i(\bone) = f_1(d_i) + f_{1x_1}(x_1) + \dots + f_{1x_p}(x_p) + \varepsilon_{1i}$ and bound them separately if $p$ covariates are collected. Essentially, this leads to a randomization for covariates balance and classical methods can be helpful \citep{liu2022adaptive, ma2020statistical}.
\end{remark}

To motivate a simple MSE surrogate function $\M(\bw)$, we further assume the maximal in-degree of network $G_n$ does not increase with sample size $n$ and omit the unknown coefficients $c_1$ and $c_2$. We define 
\begin{equation}
    \label{eq:surrogate}
    \M(\bw) := \left\{\E||P_{S_1^{\bw}} - P_G||_2^2 + \E||P_{S_0^{\bw}} - P_G||_2^2\right\} + \left\{\E{|S_1^{\bw}|}^{-1}+\E{|S_0^{\bw}|}^{-1}\right\}.
\end{equation}
Then as a natural result of Theorem \ref{theo: mse}, we have $\MSE(\hat{\tau}_{\bw}) = O\left(\M(\bw)\right)$. Given the network $G$, $\M(\bw)$ is a function of WRI weight $\bw$ and we are facing the following optimization problem:
\begin{equation}
    \label{eq: optimization}
    \bm w_\mathbf{opt} = \arg \min_{\bw} \M(\bw).
\end{equation}
The analytical calculation of $\M(\bw)$ is complicated, making classical optimization methods based on the gradients inapplicable. Fortunately, $\M(\bw)$ can be estimated through simulations, allowing us to find the best weight from a given set of candidates. This process, termed weight selection, is outlined in Algorithm \ref{alg: adaptive weight selection}. In the next subsection, we will recommend candidate weights for implementing this algorithm in practice.

\begin{remark}
    Including the maximal in-degree $\dmax$ in $\M(\bw)$ also works in most cases. However, since $\dmax$ is typically large in real networks, the new $\M(\bw)$ may become dominated by the first term. Neglecting the second term increases the risk of selecting a weight that leads to a small $|S^{\bw}|$ during the WRI process, which can significantly increase the estimator's variance. Therefore, we recommend excluding $\dmax$ from $\M(\bw)$ to achieve a more robust weight selection procedure.
\end{remark}

\begin{algorithm}[!htb]
    \caption{Weight Selection}\label{alg: adaptive weight selection}
    \DontPrintSemicolon
    \KwInput{network $G=(V, E)$, candidate set $\{\bm w_l \in \R^n_{\scriptscriptstyle \geq 0}\}_{l=1}^L$, pre-experiment number $N_p$}
    \KwOutput{optimal weight $\bm w_\mathbf{opt}$}
    $\bm d = (d_1,\dots,d_n)^\prime \leftarrow \text{in-degree vector of } G$\\
    $P_G \leftarrow \text{PMF of } \{d_i: i\in [n]\}$ \\
    \For{$l \in \{ 1,\dots, L \}$ }{
        \For{$j \in \{ 1,\dots, N_p \}$}{
            $S \leftarrow \textbf{WRI}(G,\bm w_k)$  \tcp*{isolated set}
            $S_{1} \leftarrow \textbf{CR}(S,\lfloor|S|/2\rfloor)$ \tcp*{CR conditioned on $S$}
            $S_{0} \leftarrow S\setminus S_{1}$ \tcp*{control group}
            $P_{S_1} \leftarrow \text{PMF of } \{d_i: i\in S_1\}$ \\
            $P_{S_0} \leftarrow \text{PMF of } \{d_i: i\in S_0\}$ \\
            $m_{lj} \leftarrow ||P_{S_1} - P_G||_2^2 + ||P_{S_0} - P_G||_2^2 + {|S_1|}^{-1}+{|S_0|}^{-1}$ \tcp*{MSE surrogate}
        }
        $m_l\leftarrow\frac{1}{N_p}\sum_{j=1}^{N_p}m_{lj}$ \\
    }
    $l_\text{opt}\leftarrow\arg\min_l m_l$\\
    $\bm w_\mathbf{opt}\leftarrow \bm w_{\scriptscriptstyle l_\text{opt}}$\\
    \Return{$\bm{w}_\mathbf{opt}$}
\end{algorithm}

\subsection{Recommendation for candidate weights}
In Algorithm \ref{alg: adaptive weight selection}, researchers need to specify a set of candidate weights $\{\bm w_l \in \R^n_{\scriptscriptstyle \geq 0}\}_{l=1}^L$ in advance. To effectively minimize the MSE surrogate, exploring a broader range of candidate weights would be ideal. However, the space $\R^n_{\scriptscriptstyle \geq 0}$ is too vast to fully explore, and computational burdens further limit this exploration. Therefore, it is practical to first identify a set of promising candidate weights. To illustrate this approach, we analyze a toy example. For every $l\in \R$, we define the in-degree-based weight ``$\text{degree}^l$'' as ${w}_i = d_i^l$ for $i \in [n]$. The weight $\text{degree}^0$ is commonly referred to as ``uniform'' weight. Additionally, we refer to the probability of a unit being selected into the isolated set as its inclusion probability.

\begin{example}[Path Graph $P_5$]
    \label{exam: toyexample}
    Consider an undirected network of 5 units where $A_{ij}=\mathrm{I}(0<|i-j|\leq 1)$ for all $i,j\in[5]$, as shown in Figure \ref{fig: toyexample}. This network is known as a ``path'', denoted by $P_5$ \citep{west2001introduction}. For simplicity, we ignore the randomness of complete randomization, reducing the MSE surrogate to $\M^*(\bw) = \E||P_{S^{\bw}} - P_{P_5}||_2^2 + \E|S^{\bw}|^{-1}$. We compare $\M^*(\bw)$ across six candidate weights: $\text{degree}^l$, with $l\in\{-1,0,1,2,3,4\}$. Through straightforward calculations, their respective $\M^*(\bw)$ values are 0.901, 0.820, 0.778, 0.769, 0.772, 0.778.
\end{example}

\begin{figure}[!htb]
    \centering
    \includegraphics[width=0.95\textwidth]{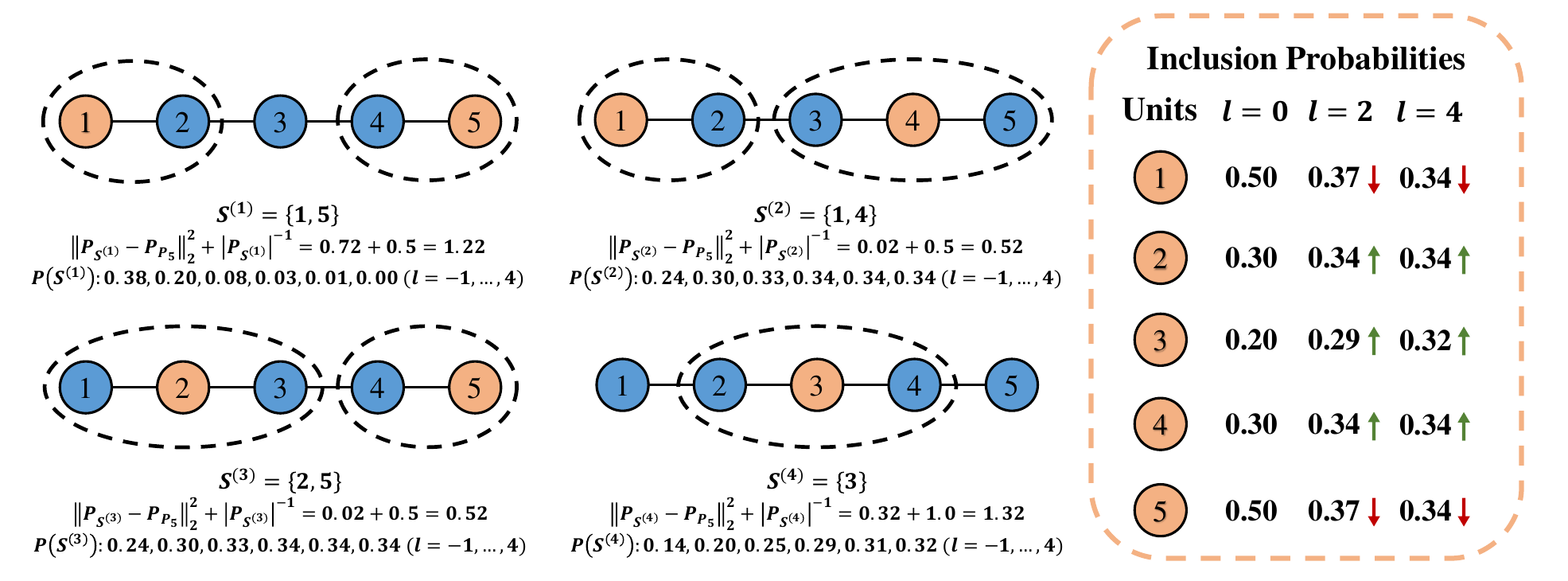}
    \caption{A toy example illustrating the role of weights $\text{degree}^l$, with $l\in\{-1,0,1,2,3,4\}$. The network is a chain with 5 units. The sets $\{S^{(m)},m \in [4]\}$ represent 4 possible isolated sets generated by the WRI design. $P(S^{(m)})$ indicate the probabilities of units being sampled into $S^{(m)}$ for $l=-1,\dots,4$. $||P_{S^{(m)}} - P_{P_5}||_2^2 + |S^{(m)}|^{-1}$ shows the corresponding $\M^*(\bw)$ value given the set $S^{(m)}$. The orange units represent isolated units and the black dashed ovals circles clusters. On the right, units' inclusion probabilities of $l=0,2,4$ are summarized.}
    \label{fig: toyexample}
\end{figure}

As Example \ref{exam: toyexample} demonstrates, $\M^*(\bw)$ for weights $\text{degree}^l$ decreases initially and then increases as $l$ grows, which we believe reflects a common pattern in general networks. The intuition behind this is that as $l$ increases, the inclusion probabilities of units with high in-degrees, which are typically lower than those of units with low in-degrees, are improved. This adjustment leads to more even inclusion probabilities across units, as shown in the right part of Figure \ref{fig: toyexample}. Consequently, the first term of $\M^*(\bw)$ decreases. However, the inclusion of units with high in-degrees is often associated with a smaller valid sample size $|S^{\bw}|$, which ultimately causes an increase in the second term of $\M^*(\bw)$. As a result, we expect $\M^*(\bw)$ for weights $\text{degree}^l$ to reach a local minimum when $l$ is at a moderate value. 

Therefore, we recommend including $\{\text{degree}^l, l\in\{-1,0,1,2,3,4\}\}$ in the candidate set. Additionally, inspired by \cite{ugander2023randomized}, we also include $\{\text{spectral}^l, l\in\{-1,0,1,2,3,4\}\}$ in the candidate set, where the spectral weight is defined as the eigenvector associated with the spectral radius of the 2-order adjacency matrix of $G$. The concept of the 2-order adjacency matrix is detailed in Section \ref{app: notation}. The definition of $\text{spectral}^l$ is analogous to that of $\text{degree}^l$. Since these weights are derived solely from the network structure, we describe our weight selection procedure as network-adaptive. In Figure \ref{fig: adaptive}, we demonstrate that adaptive selection based on the MSE surrogate from our recommended candidate set effectively reduces $\MSE(\hat{\tau}_{\bw})$ compared to the original ``uniform'', ``degree'', and ``spectral'' weights. In practice, if computational resources allow, one could include additional candidate weights to further enhance performance.

\begin{figure}[!htb]
    \centering
    \includegraphics[width=0.95\textwidth]{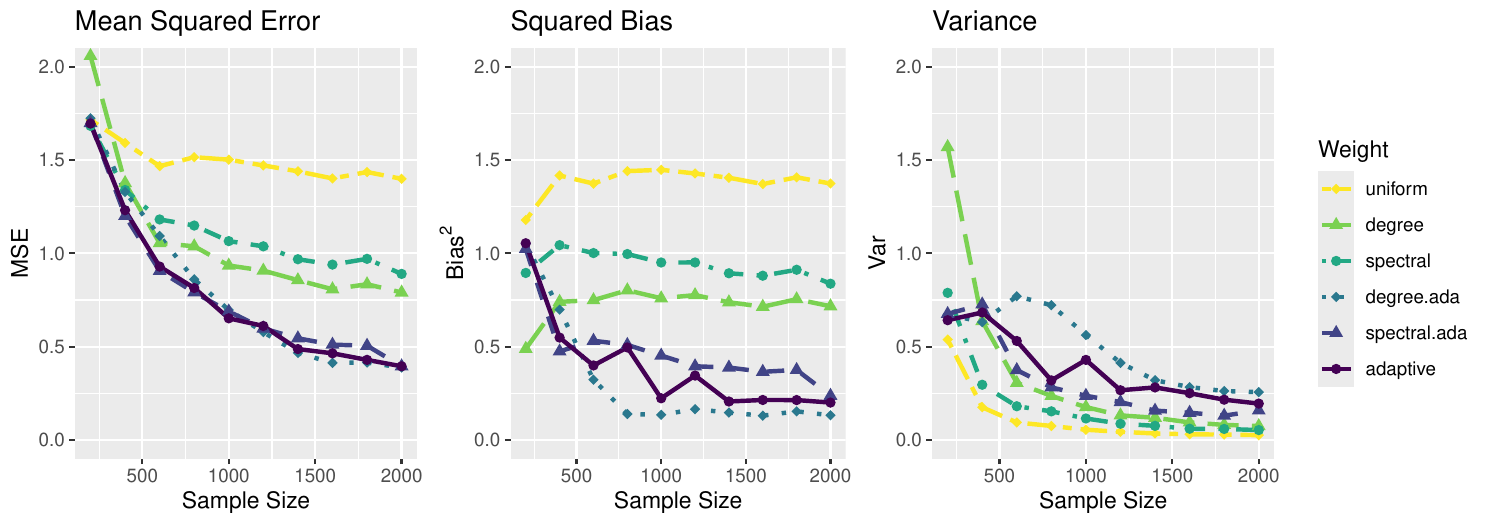}
    \caption{The MSE, Squared Bias and Variance of $\hat{\tau}_{\bw}$ with with different weighting strategies versus the Sample Size. The ``uniform'', ``degree'' and ``spectral'' represent corresponding weights. The ``degree.ada'', ``spectral.ada'' and ``adaptive'' represent the adaptive selection based on the MSE surrogate from the candidate set $\{\text{degree}^l\}_{l=-1}^4$, $\{\text{spectral}^l\}_{l=-1}^4$, and the union set of them. The setting of network model and potential outcome model is same as that in Section \ref{subsec: conjecture}. 1000 replications are conducted.}
    \label{fig: adaptive}
\end{figure}

\section{Numerical studies}
\label{sec: simulation}

\subsection{Setup}
\label{subsec: setup of simulation}
\paragraph*{Networks}
To compare with exsiting methods, all networks considered here are undirected. We examine five popular network models: the Barab\'asi-Albert model (BA), the random geometric model (RG), the small-world model (SW), the Erd\H{o}s-R\'enyi model (ER), and the stochastic block model (SBM). Under the finite population setting, we pre-generate 5 specific networks from the corresponding models. In our settings, the BA model simulates a network with both low-degree and extremely high-degree units, the RG model stimulates a sparse newtwork and the SW model stimulates a dense network \citep{albert2002statistical}. The ER model and SBM model are two popular choices in recent literature \citep{cortez2024combining} and their basic properties are listed in Table \ref{tab: network}.

\begin{table}[!htb]
\centering
\caption{Basic properties of 5 specific networks. $m$ denotes the number of edges, $\dmin$ and $\dmax$ denote the minimum and maximum degrees, $\overline{d}$ denotes the average degree, $\dmax^{(2)}$ denotes the maximum degree of the 2-order adjacency matrix, $\overline{d^{(2)}}$ denotes the average degree of the 2-order adjacency matrix, $\overline{\text{dist}}$ denotes the average shortest path length, and ``diam'' denotes the diameter.}
\label{tab: network}
\begin{tabular}{lllllllllll}
\toprule
network & & $n$ & $m$ & $\dmin$ & $\dmax$ & $\overline{d}$ & $\dmax^{(2)}$ & $\overline{d^{(2)}}$ & $\overline{\text{dist}}$ & diam \\
\midrule
BA & & 1200  & 3594 & 3 & 78      & 5.990     & 627           & 73.07           & 3.621         & 6        \\
RG & & 1200  & 4048 & 1 & 15      & 6.747     & 37            & 17.08           & 18.78         & 45       \\
SW & & 1200  & 6000 & 2 & 20      & 10.00     & 204           & 101.8           & 3.358         & 6        \\
ER & & 1200  & 4141 & 1 & 14      & 6.902     & 115           & 52.95           & 3.895         & 8        \\
SBM & & 1200  & 3160 & 1 & 14      & 5.267     & 77            & 29.72           & 7.844         & 20 \\
\bottomrule    
\end{tabular}
\end{table}

\paragraph*{Potential outcome models}
Three models are considered. The first one is based on \cite{ugander2023randomized} with modifications to account for the heterogeneity in the direct effect parameter $\delta_i$ and spillover effect parameter $\gamma_i$:
\begin{equation}
    \label{eq: ugander}
    \begin{aligned}
    &Y_{i}(\bzero)=(a+b h_i+\sigma \epsilon_i) \frac{d_i}{\bar{d}},\\
    &Y_{i}(\bZ)= Y_i(\bzero)\left(1+\delta_i Z_i+\gamma_i\frac{\sum_{j=1}^n A_{ij}Z_j}{d_i}\right).\end{aligned}
\end{equation}
In this model, $a$ is a baseline effect. $(h_i)_{i=1}^n$, defined as the eigenvector associated with the second smallest eigenvalue of the normalized network Laplacian matrix $D^{-1}L$, captures possible homophily in the network. $\varepsilon_i\iidsim N(0,1)$ is a random perturbation of the baseline effect. we set $(a,b,\sigma)=(1,0.5,0.1)$. $\delta_i\iidsim N(0.5,0.01)$ and $\gamma_i\iidsim N(1,0.01)$ control the direct and spillover effects, respectively. Under the finite population setting, we generate the $\varepsilon_i$'s, $\delta_i$'s and $\gamma_i$'s in advance and keep them consistent across all experiments. This model satisfies the assumption of full neighborhood interference, and the potential outcomes $Y_i(\bone)$'s and $Y_i(\bzero)$'s are correlated with the degree $d_i$, satisfying the Assumption \ref{asu: potential outcomes decomposition}.

The second model is a complex linear model based on \cite{Leung2022}. Unlike their setting, we truncate the infinite summation at 10 for simplicity:
\begin{equation}
    \bY(\bZ)=\frac\alpha{1-\beta}\bone+\gamma\bZ+\gamma\sum_{j=1}^{10}\beta^j\tilde{\bA}^{j}\bZ+\sum_{j=0}^{10}{\beta}^j\tilde{\bA}^j\bm{\varepsilon},
\end{equation}
where $\bY(\bZ)=(Y_i(\bZ))_{i=1}^n$, $\bm{\varepsilon}=(\varepsilon_i)_{i=1}^n$ , and $\tilde{\bA}$ is the row-normalized version of $\bA$ (each row divided by its sum). The third term indicates that the impact of treatments assigned to the $j$-neighborhood is exponentially down-weighted by $\beta^j$. We set $(\alpha,\beta,\gamma)=(-1,0.8,1)$ and generate $\varepsilon_i \iidsim N(0,1)$ in advance. This model violates the full neighborhood interference assumption, so it can be used as a robustness check for the methods.

The third model is a complex contagion model based on \cite{Leung2022}. The dynamic discrete-time process is initialized at period $t=0$ with a binary response vector ${\bY}^0\in\{0,1\}^n$. For a given $\bZ$, the model updates as follows:
\begin{equation}
    Y_i^{t+1}(\bZ)=\mathbf{1}\left\{\alpha+\beta\frac{\sum_j A_{ij}Y_j^{t}}{\sum_jA_{ij}}+\delta\frac{\sum_jA_{ij}Z_j}{\sum_jA_{ij}}+Z_i\gamma+\varepsilon_i>0\right\}.
\end{equation}
The process continues until the first period $T$ where ${\bY}^{T+1}(\bZ)={\bY}^{T}(\bZ)$, and we define the potential outcomes as $\bY(\bZ) = {\bY}^{T}(\bZ)$.
We set $(\alpha,\beta,\delta,\gamma)=(-1,1.5,1,1)$ and generate ${Y}^{0}_i\overset{iid}{\sim}\Ber(0.5)$, $\varepsilon_i \iidsim N(0,1)$ in advance. This model violates the full neighborhood interference assumption and potential outcomes take values only from $\{0, 1\}$, which are significantly different from those in the previous settings. The potential outcomes of three models under the BA network are shown in \ref{fig: outcomes}.

\begin{figure}[!htb]
    \centering
    \includegraphics[width=0.95\textwidth]{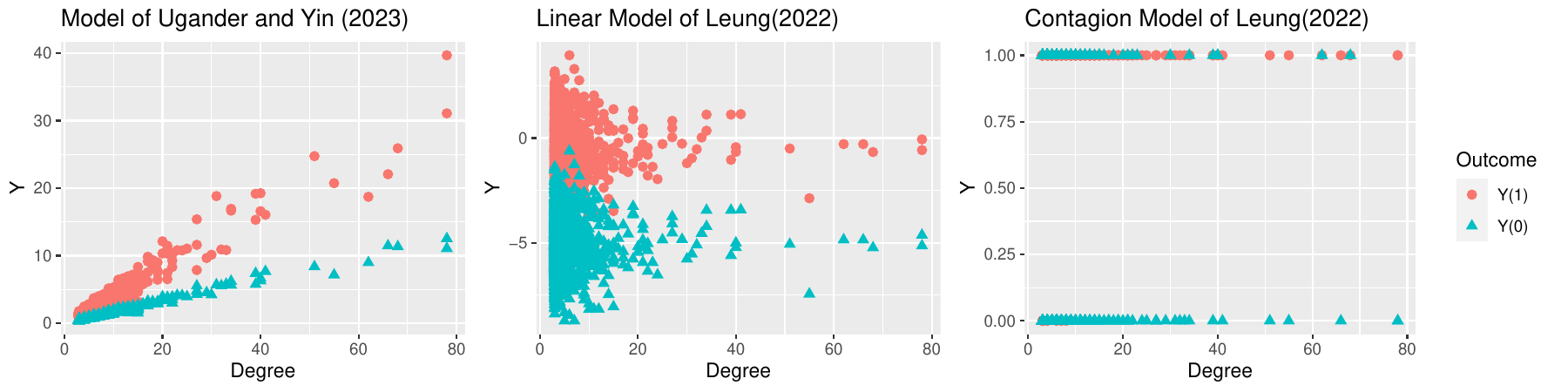}
    \caption{The potential outcomes of 3 models under BA network.}
    \label{fig: outcomes}
\end{figure}

Throughout the entire simulation, we focus exclusively on estimating the total treatment effect (TTE), which is defined by $\TTE = \frac{1}{n}\sum_{i=1}^n (Y_i(1)-Y_i(0))$. We conduct 1000 replications for each combination of network and potential outcome model.

\subsection{Results}
We compare our methods with nine existing ones (DESIGN+estimator) from the literature. The first six are based on IPW-type estimators: naive Bernoulli design with Horvitz-Thompson estimator (BER+ht) \citep{Leung2022}, graph clustering randomization design with HT estimator (GCR+ht) \citep{UganderKarrerBackstrom2013}, randomized graph clustering randomization design with HT estimator (RGCR+ht) \citep{ugander2023randomized}, naive Bernoulli design with H\'ajek estimator (BER+hajek) \citep{gao2023causal}, GCR design with H\'ajek estimator (GCR+hajek) \citep{UganderKarrerBackstrom2013}, and RGCR design with H\'ajek estimator (RGCR+hajek) \citep{ugander2023randomized}. The last three are based on difference-in-means-type estimators: naive Bernoulli design with naive difference-in-means (BER+dim) estimator, causal clustering design with naive difference-in-means estimator (CC+dim) \citep{viviano2023causal}, independent-set design with ordinary least squares estimator (IS+ols) \citep{cai2023independent}.

For our methods, we consider 4 candidates: random isolation design with restricted difference-in-means estimator (RI+rdim), RI design with restricted matched estimator (RI+rmat), adaptive weighted random isolation design with restricted difference-in-means estimator (AWRI+rdim) and AWRI design with restricted matched estimator (AWRI+rmat).

\begin{table}[!htb]
    \caption{The MSE of 13 methods under 5 networks. Potential outcomes are generated based on Equation \ref{eq: ugander} \citep{ugander2023randomized}.}
    \label{tab: ugander}
    \centering
    \resizebox{\textwidth}{!}{
    \begin{tabular}{lllllllllllllllllllll}
            \toprule
    DESIGN+estimator &  & \multicolumn{3}{c}{BA} &           & \multicolumn{3}{c}{RG} &  & \multicolumn{3}{c}{SW} &  & \multicolumn{3}{c}{ER} &           & \multicolumn{3}{c}{SBM} \\
    &  & MSE  & Bias$^2$  & Var &           & MSE  & Bias$^2$  & Var &  & MSE  & Bias$^2$  & Var &  & MSE  & Bias$^2$  & Var &           & MSE  & Bias$^2$  & Var  \\ 
    \cline{1-1} \cline{3-5} \cline{7-9} \cline{11-13} \cline{15-17} \cline{19-21} 
    BER+ht & & 7.474 & 0.322 & 7.152 & & 63.31 & 0.124 & 63.19 & & 98.64 & 0.000 & 98.64 & & 10.48 & 0.052 & 10.43 & & 4.154 & 0.011 & 4.143 \\  
    GCR+ht & & 4.429 & 0.200 & 4.229 & & 0.616 & 0.000 & 0.616 & & 36.44 & 0.038 & 36.40 & & 9.938 & 0.009 & 9.928 & & 3.367 & 0.001 & 3.366 \\  
    RGCR+ht & & 2.407 & 0.003 & 2.403 & & 0.154 & 0.000 & 0.154 & & 0.245 & 0.006 & 0.239 & & 0.259 & 0.011 & 0.248 & & 0.275 & 0.006 & 0.269 \\  
    BER+hajek & & 0.412 & 0.215 & 0.196 & & 0.454 & 0.118 & 0.336 & & 0.598 & 0.368 & 0.231 & & 0.361 & 0.061 & 0.300 & & 0.259 & 0.019 & 0.240 \\  
    GCR+hajek & & 0.330 & 0.153 & 0.177 & & 0.029 & 0.000 & 0.029 & & 0.239 & 0.054 & 0.186 & & 0.236 & 0.022 & 0.213 & & 0.116 & 0.004 & 0.112 \\  
    RGCR+hajek & & 0.532 & 0.010 & 0.522 & & 0.022 & 0.000 & 0.022 & & 0.032 & 0.005 & 0.028 & & \textbf{0.027} & 0.005 & 0.022 & & \textbf{0.021} & 0.002 & 0.019 \\  
    \\  
    BER+dim & & 1.057 & 1.044 & 0.013 & & 0.977 & 0.975 & 0.002 & & 1.011 & 1.010 & 0.001 & & 0.962 & 0.960 & 0.002 & & 1.000 & 0.998 & 0.002 \\  
    CC+dim & & 0.982 & 0.512 & 0.470 & & 0.261 & 0.022 & 0.239 & & 0.697 & 0.585 & 0.113 & & 0.594 & 0.525 & 0.068 & & 0.561 & 0.211 & 0.350 \\
    IS+ols & & 0.339 & 0.328 & 0.011 & & 0.231 & 0.190 & 0.042 & & 0.107 & 0.072 & 0.035 & & 0.193 & 0.165 & 0.028 & & 0.213 & 0.192 & 0.021 \\  
    RI+rdim & & 0.299 & 0.282 & 0.016 & & 0.099 & 0.082 & 0.017 & & 0.113 & 0.090 & 0.023 & & 0.184 & 0.166 & 0.019 & & 0.197 & 0.181 & 0.016 \\  
    RI+rmat & & 0.287 & 0.280 & 0.008 & & 0.086 & 0.083 & 0.003 & & 0.091 & 0.083 & 0.007 & & 0.169 & 0.165 & 0.004 & & 0.190 & 0.188 & 0.003 \\  
    AWRI+rdim & & 0.191 & 0.024 & 0.167 & & 0.029 & 0.008 & 0.021 & & 0.046 & 0.010 & 0.036 & & 0.058 & 0.026 & 0.031 & & 0.073 & 0.050 & 0.023 \\  
    AWRI+rmat & & \textbf{0.106} & 0.020 & 0.086 & & \textbf{0.012} & 0.008 & 0.004 & & \textbf{0.018} & 0.008 & 0.011 & & 0.031 & 0.025 & 0.006 & & 0.052 & 0.047 & 0.004 \\     
    \bottomrule
    \end{tabular}
    }
\end{table}

Table \ref{tab: ugander} presents the MSE of 13 methods across 5 network types under the first potential outcome model. Overall, methods based on IPW-type estimators exhibit very small bias but large variance, while those based on DIM-type estimators demonstrate the opposite trend. Our method ``AWRI+rmat'' outperforms others in the graph BA, RG and SW, whereas ``RGCR+hajek'' performs best in the network ER and SBM. 

Among all methods based on IPW-type estimators, the HT estimators have the worst performance due to their unacceptably high variance, despite being unbiased theoretically. H\'ajek estiamtors, which can be viewd as a regularized version of HT estimators, significantly reduce the variance at the cost of introducing a small bias. Additionally, by refining the design from BER to GCR to RGCR, the variances of both HT and H\'ajek estimators are further reduced.

Among all methods based on DIM-type estimators, the naive Bernoulli design performs worst due to its large bias. The method ``BER+dim'' is unacceptable bacause it is inconsistent with the target effect. The CC design addresses this issue to some extent, but it requires many prior information to tune parameters, making it difficult to implement in practice. Additionally, this design will not be included in subsequent comparisons due to its high computational burden. The IS design performs well in this scenario, but it breaks down in the complex linear model as shown in Tabel \ref{tab: leung-linear}. This is because the method ``IS+ols'' heavily relies on the correct specification of the outcome model, assuming the direct and spillover effects are additive, which is usually violated in practice. The RI designs show reasonable performance but still suffer from non-negligible biases. The AWRI designs effectively address this issue, with a small increase in variance. finally, substituting the restricted DIM estimator to the restricted matched estimator can further reduce the MSE in this case.

Table \ref{tab: leung-linear} and Table \ref{tab: leung-contagion} tell a similar story. Under the models of \cite{Leung2022}, potential outcomes do not strongly correlated with the network degree, as illustrated in Figure \ref{fig: outcomes}. Thus it is not expected that the adaptive weight selection can enhance the performance of WRI. Even so, methods based on AWRI design are still comparable to the best ones in the same setting. Given its simplicity and robustness, it is preferable in practice.

\begin{table}[!htb]
    \caption{The MSE of 12 methods under 5 networks. Potential outcomes are generated based on complex linear model \citep{Leung2022}.}
    \label{tab: leung-linear}
    \centering
    \resizebox{\textwidth}{!}{
    \begin{tabular}{lllllllllllllllllllll}
            \toprule
    DESIGN+estimator &  & \multicolumn{3}{c}{BA} &           & \multicolumn{3}{c}{RG} &  & \multicolumn{3}{c}{SW} &  & \multicolumn{3}{c}{ER} &           & \multicolumn{3}{c}{SBM} \\
    &  & MSE  & Bias$^2$  & Var &           & MSE  & Bias$^2$  & Var &  & MSE  & Bias$^2$  & Var &  & MSE  & Bias$^2$  & Var &           & MSE  & Bias$^2$  & Var  \\ 
    \cline{1-1} \cline{3-5} \cline{7-9} \cline{11-13} \cline{15-17} \cline{19-21} 
BER+ht & & 15.53 & 0.083 & 15.45 & & 36.48 & 0.002 & 36.48 & & 109.2 & 0.123 & 109.0 & & 33.11 & 0.000 & 33.11 & & 7.766 & 0.002 & 7.765 \\  
GCR+ht & & 11.51 & 0.052 & 11.46 & & 1.211 & 0.001 & 1.209 & & 140.3 & 0.057 & 140.2 & & 21.17 & 0.000 & 21.17 & & 7.270 & 0.058 & 7.212 \\  
RGCR+ht & & 2.163 & 0.000 & 2.163 & & 0.294 & 0.004 & 0.290 & & 0.786 & 0.000 & 0.786 & & 0.584 & 0.024 & 0.560 & & 0.714 & 0.018 & 0.696 \\  
BER+hajek & & 0.215 & 0.000 & 0.215 & & 0.940 & 0.000 & 0.939 & & 1.771 & 0.003 & 1.770 & & 0.614 & 0.001 & 0.614 & & 0.400 & 0.000 & 0.400 \\  
GCR+hajek & & 0.109 & 0.000 & 0.109 & & 0.072 & 0.000 & 0.072 & & 0.499 & 0.000 & 0.499 & & 0.234 & 0.001 & 0.233 & & 0.153 & 0.000 & 0.153 \\  
RGCR+hajek & & 0.060 & 0.000 & 0.060 & & \textbf{0.055} & 0.000 & 0.055 & & 0.109 & 0.000 & 0.109 & & 0.056 & 0.000 & 0.055 & & 0.076 & 0.000 & 0.076 \\
\\    
BER+dim & & 11.93 & 11.92 & 0.005 & & 11.07 & 11.06 & 0.008 & & 12.25 & 12.25 & 0.004 & & 12.00 & 11.99 & 0.005 & & 11.61 & 11.61 & 0.006 \\
IS+ols & & 6.124 & 6.061 & 0.063 & & 2.615 & 2.452 & 0.163 & & 6.885 & 6.728 & 0.157 & & 6.310 & 6.221 & 0.089 & & 5.514 & 5.427 & 0.087 \\  
RI+rdim & & 0.056 & 0.000 & 0.056 & & 0.061 & 0.000 & 0.061 & & \textbf{0.090} & 0.000 & 0.090 & & 0.057 & 0.000 & 0.057 & & 0.048 & 0.000 & 0.048 \\  
RI+rmat & & \textbf{0.053} & 0.000 & 0.053 & & 0.070 & 0.000 & 0.070 & & 0.094 & 0.000 & 0.094 & & \textbf{0.055} & 0.000 & 0.055 & & \textbf{0.045} & 0.000 & 0.045 \\  
AWRI+rdim & & 0.064 & 0.000 & 0.064 & & 0.068 & 0.000 & 0.068 & & 0.105 & 0.000 & 0.105 & & 0.059 & 0.000 & 0.059 & & 0.054 & 0.000 & 0.053 \\  
AWRI+rmat & & 0.069 & 0.000 & 0.069 & & 0.069 & 0.000 & 0.069 & & 0.105 & 0.000 & 0.105 & & 0.064 & 0.000 & 0.064 & & 0.054 & 0.000 & 0.054 \\      
    \bottomrule
    \end{tabular}
    }
\end{table}

\begin{table}[!htb]
    \caption{The MSE of 12 methods under 5 networks. Potential outcomes are generated based on complex contagion model \citep{Leung2022}.}
    \label{tab: leung-contagion}
    \centering
    \resizebox{\textwidth}{!}{
    \begin{tabular}{lllllllllllllllllllll}
            \toprule
    DESIGN+estimator &  & \multicolumn{3}{c}{BA} &           & \multicolumn{3}{c}{RG} &  & \multicolumn{3}{c}{SW} &  & \multicolumn{3}{c}{ER} &           & \multicolumn{3}{c}{SBM} \\
    &  & MSE  & Bias$^2$  & Var &           & MSE  & Bias$^2$  & Var &  & MSE  & Bias$^2$  & Var &  & MSE  & Bias$^2$  & Var &           & MSE  & Bias$^2$  & Var  \\ 
    \cline{1-1} \cline{3-5} \cline{7-9} \cline{11-13} \cline{15-17} \cline{19-21} 
BER+ht & & 5.136 & 0.006 & 5.130 & & 2.079 & 0.001 & 2.079 & & 10.29 & 0.016 & 10.27 & & 1.701 & 0.001 & 1.700 & & 0.462 & 0.000 & 0.462 \\  
GCR+ht & & 13.41 & 0.009 & 13.40 & & 0.076 & 0.000 & 0.076 & & 4.581 & 0.000 & 4.580 & & 1.097 & 0.000 & 1.097 & & 0.307 & 0.000 & 0.307 \\  
RGCR+ht & & 0.125 & 0.000 & 0.125 & & 0.014 & 0.000 & 0.014 & & 0.034 & 0.000 & 0.034 & & 0.032 & 0.000 & 0.032 & & 0.033 & 0.000 & 0.033 \\  
BER+hajek & & 0.020 & 0.000 & 0.020 & & 0.050 & 0.000 & 0.050 & & 0.144 & 0.001 & 0.144 & & 0.049 & 0.000 & 0.049 & & 0.026 & 0.000 & 0.026 \\  
GCR+hajek & & 0.008 & 0.000 & 0.008 & & 0.002 & 0.000 & 0.002 & & 0.041 & 0.000 & 0.041 & & 0.020 & 0.000 & 0.020 & & 0.007 & 0.000 & 0.007 \\  
RGCR+hajek & & \textbf{0.003} & 0.000 & 0.003 & & \textbf{0.001} & 0.000 & 0.001 & & 0.008 & 0.000 & 0.008 & & \textbf{0.004} & 0.000 & 0.004 & & \textbf{0.002} & 0.000 & 0.002 \\
\\    
BER+dim & & 0.209 & 0.208 & 0.000 & & 0.177 & 0.176 & 0.000 & & 0.254 & 0.254 & 0.000 & & 0.203 & 0.203 & 0.000 & & 0.207 & 0.207 & 0.000 \\
IS+ols & & 0.058 & 0.053 & 0.005 & & 0.017 & 0.009 & 0.008 & & 0.068 & 0.055 & 0.013 & & 0.039 & 0.032 & 0.007 & & 0.040 & 0.035 & 0.004 \\  
RI+rdim & & 0.004 & 0.000 & 0.004 & & 0.003 & 0.000 & 0.003 & & \textbf{0.007} & 0.000 & 0.007 & & 0.004 & 0.000 & 0.004 & & 0.003 & 0.000 & 0.002 \\  
RI+rmat & & 0.004 & 0.000 & 0.004 & & 0.003 & 0.000 & 0.003 & & 0.007 & 0.000 & 0.007 & & 0.004 & 0.000 & 0.004 & & 0.003 & 0.000 & 0.003 \\  
AWRI+rdim & & 0.005 & 0.000 & 0.005 & & 0.004 & 0.001 & 0.003 & & 0.007 & 0.000 & 0.007 & & 0.005 & 0.000 & 0.005 & & 0.003 & 0.000 & 0.003 \\  
AWRI+rmat & & 0.005 & 0.000 & 0.005 & & 0.004 & 0.001 & 0.003 & & 0.008 & 0.000 & 0.008 & & 0.005 & 0.000 & 0.005 & & 0.003 & 0.000 & 0.003 \\ 
    \bottomrule
    \end{tabular}
    }
\end{table}

\subsection{Conjecture on the consistency}
\label{subsec: conjecture}
In previous sections, we did not provide any theoretical guarantees for the convergence of methods based on the AWRI design. Here, we explore these aspects through simulations. Specifically, we compare Mean squared error, Squared bias and variance of eight methods (IS+ols, BER+hajek, GCR+hajek, RGCR+hajek, RI+rdim, RI+rmat, AWRI+rdim and AWRI+rmat) as sample size increases.

The detailed settings are as follows. Networks are generated from the BA network with sample size $n=200,400,\dots,2000$. Potential outcomes are generated based on the model of \cite{ugander2023randomized} (Equation \ref{eq: ugander}). All parameters remain consistent with the settings in Section \ref{subsec: setup of simulation}, except that spillover effect parameters $\gamma_i\iidsim N(3,0.01)$. 1000 replications are conducted. 

\begin{figure}[!htb]
    \centering
    \includegraphics[width=0.95\textwidth]{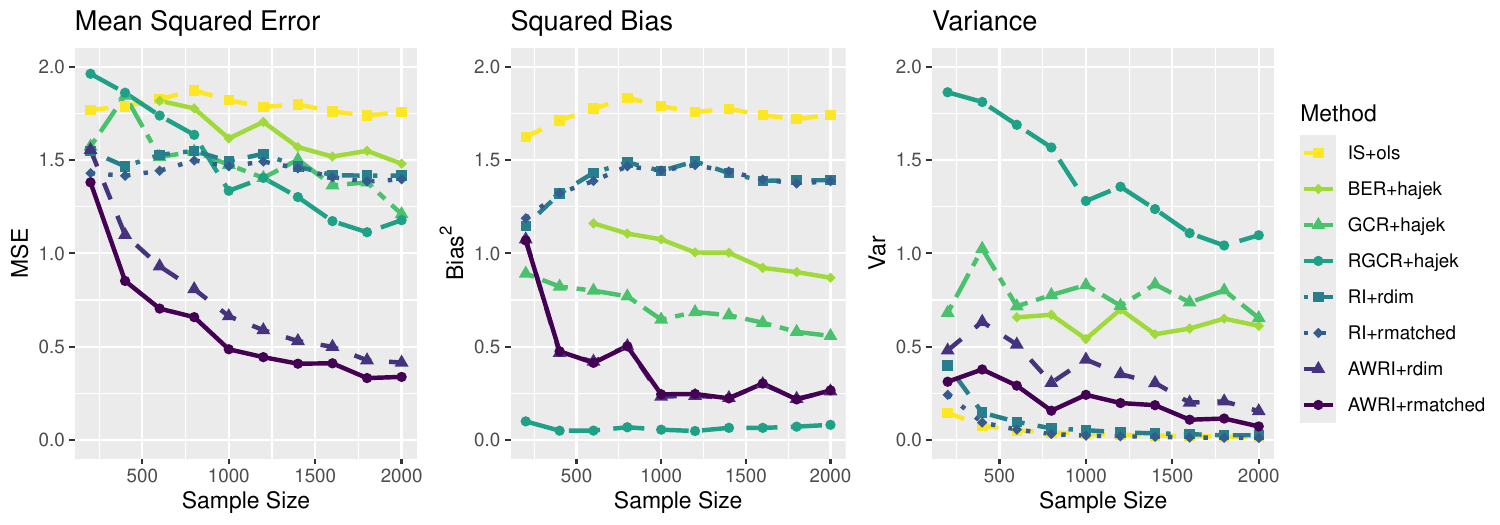}
    \caption{The MSE, Squared Bias, and Variance versus the Sample Size for 8 Methods.}
    \label{fig: ugander-ba-8methods}
\end{figure}

The results are shown in Figure \ref{fig: ugander-ba-8methods}. The ``IS+ols'', ``RI+rdim'' and ``RI+rmat'' methods do not exhibit clear convergence due to their uncontrollable bias. The three methods based on H\'ajek estimators converge at a slow rate. In contrast, the methods based on the AWRI design converge more rapidly. Therefore, we can reasonably expect that the proposed methods possess consistency under some conditions.

\section{Conclusion} 
\label{sec: conclusion}

In this paper, we introduce a novel randomized design called adaptive weighted random isolation (AWRI), which can sample an isolated set where causal effects can be estimated as if no interference exists between the isolated units. We pair this design with a restricted difference-in-means estimator, which is typically not favored in classical survey sampling literature. However, in the presence of complex interference, facing the task of estimating total treatment effects (TTE), this method outperforms the best existing approaches. Due to its simplicity and interpretability, we recommend this method for practical network experiments.

There are several directions for future research. First, this method can be easily extended to estimate spillover effects, direct effects, and other effects of interest in the presence of interference. And the concept of isolated units naturally aligns with that of focal units in \cite{AtheyEcklesImbens2018a}, enabling the extension of corresponding randomization inference \citep{basse2019randomization, puelz2022graph, basse2024randomization}. Second, more rigorous theoretical exploration is needed. Under network interference and finite population settings, deriving meaningful theoretical results is extremely challenging, and only limited work provided insights into this area \citep{kojevnikov2021limit, Leung2022, gao2023causal, viviano2023causal}. Third, in practice, social networks may be weighted or observed with noise \citep{egami2021spillover, hardy2019estimating, young2020bayesian}. More critically, they may even be unobserved entirely, making multi-stage designs a valuable consideration in such cases \citep{li2021causal, yu2022estimating, cortez2022staggered, cortez2024combining}. Finally, the intuition behind adaptive weight selection also applies to other methods such as RGCR, which assigns every unit a weight to adjust their exposure probabilities. The challenge lies in the complexity of new mean squared error (MSE) surrogates induced by Horvitz-Thompson (HT) and H\'ajek estimators, and optimizing weights can be computationally expensive due to the reliance on Monte Carlo methods for estimating exposure probabilities \citep{ugander2023randomized}, necessitating effective approximation techniques.

\newpage 
\bibliographystyle{ecta}
\bibliography{bib}

\appendix

\renewcommand{\theproposition}{S\arabic{proposition}}
\renewcommand{\theexample}{S\arabic{example}}
\renewcommand{\thetable}{S\arabic{table}}
\renewcommand{\theequation}{S\arabic{equation}}
\renewcommand{\thelemma}{S\arabic{lemma}}
\renewcommand{\thesection}{S\arabic{section}}
\renewcommand{\thetheorem}{S\arabic{theorem}}
\renewcommand{\thecorollary}{S\arabic{corollary}}
\renewcommand{\theremark}{S\arabic{remark}}

\setcounter{equation}{0}
\renewcommand {\theequation} {S\arabic{equation}}
\setcounter{lemma}{0}
\renewcommand {\thelemma} {S\arabic{lemma}}
\setcounter{definition}{0}
\renewcommand {\thedefinition} {S\arabic{definition}}
\setcounter{example}{0}
\renewcommand {\theexample} {S\arabic{example}}
\setcounter{proposition}{0}
\renewcommand {\theproposition} {S\arabic{proposition}}
\setcounter{corollary}{0}
\renewcommand {\thecorollary} {S\arabic{corollary}}

\newpage

\begin{center}
\Large\bfseries{Supplementary Material}
\end{center}

\section{Notation}
\label{app: notation}
Given a network $G=(V,E)$, let $G_2 = (V, E_2)$ denote the ``squared'' network, \ie, with the same unit set $V$, and an edge $(i, j) \in E_2$ if and only if there exists a unit $l$ such that $(i,l)\in E$ and $(l,j)\in E$. And we call its corresponding adjacent matrix as 2-order adjacent matrix of $G$. In an undirected network, the unnormalized network Laplacian matrix is defined by $\bL = \bD - \bA$, where $\bD$ is the diagonal degree matrix.

\section{Proofs} 
\label{app: proofs}

\subsection{Proof of Theorem \ref{theo: mse}}
\label{app: proof of theorem 4.1}
\begin{proof}[Proof of Theorem \ref{theo: mse}]
    First, we have
\begin{equation*}
    \begin{aligned}
    \MSE(\hat{\tau}_{\bw})  
    &= \E(\hat{\tau}_{\bw}-\tau)^2 \\
    &= \E((\hat{\tau}_{1,\bw}-\hat{\tau}_{0,\bw}) - (\tau_1-\tau_0))^2 \\ 
    &\leq 2\E(\hat{\tau}_{1,\bw}-\tau_1)^2 + 2\E(\hat{\tau}_{0,\bw}-\tau_0)^2 \\
    &:= 2I_1 + 2I_0,  
\end{aligned}
\end{equation*}
where $I_1= \E(\hat{\tau}_{1,\bw}-\tau_1)^2$ and $I_0= \E(\hat{\tau}_{0,\bw}-\tau_0)^2$. 

By Assumption \ref{asu: potential outcomes decomposition}, $I_1$ can be written as
$$
\begin{aligned}
    I_1
    & = \E\left(\frac{1}{|S_1^{\bw}|}\sum_{i\in S_1^{\bw}}Y_i(1) - \frac{1}{n}\sum_{i\in [n]}Y_i(1)\right)^2 \\
    & = \E\left(\frac{1}{|S_1^{\bw}|}\sum_{i\in S_1^{\bw}}(f_1(d_i)+\varepsilon_{1i}) - \frac{1}{n}\sum_{i\in [n]}(f_1(d_i)+\varepsilon_{1i})\right)^2 \\
    & = \E\left( \left(\frac{1}{|S_1^{\bw}|}\sum_{i\in S_1^{\bw}}f_1(d_i) - \frac{1}{n}\sum_{i\in [n]}f_1(d_i)\right) + \left(\frac{1}{|S_1^{\bw}|}\sum_{i\in S_1^{\bw}}\varepsilon_{1i} - \frac{1}{n}\sum_{i\in [n]}\varepsilon_{1i}\right)\right)^2. \\
\end{aligned}
$$
Let $P_{S_1^{\bw}}$, $P_{S_0^{\bw}}$ and $P_G$ indicate probabilistic mass functions (PMF) of $\{d_i, i\in S_1^{\bw}\}$, $\{d_i,i\in S_0^{\bw}\}$ and $\{d_i, i \in [n]\}$, respectively. Rewrite the first term,
$$
\begin{aligned}
    I_1
    & = \E\left( \left(\sum_{d=0}^{d_\text{max}}f_1(d)P_{S_1^{\bw}}(d) - \sum_{d=0}^{d_\text{max}}f_1(d)P_G(d)\right) + \E\left(\frac{1}{|S_1^{\bw}|}\sum_{i\in S_1^{\bw}}\varepsilon_{1i} - \frac{1}{n}\sum_{i\in [n]}\varepsilon_{1i}\right)\right)^2 \\ 
    & \leq 2\E\left(\sum_{d=0}^{d_\text{max}}f_1(d)P_{S_1^{\bw}}(d) - \sum_{d=0}^{d_\text{max}}f_1(d)P_G(d)\right)^2 +  2\E\left(\frac{1}{|S_1^{\bw}|}\sum_{i\in S_1^{\bw}}\varepsilon_{1i} - \frac{1}{n}\sum_{i\in [n]}\varepsilon_{1i}\right)^2\\
    & = 2\E\left(\sum_{d=0}^{d_\text{max}}f_1(d)(P_{S_1^{\bw}}(d)-P_G(d))\right)^2 + 2\E\left(\frac{1}{|S_1^{\bw}|}\sum_{i\in S_1^{\bw}}\varepsilon_{1i} - \frac{1}{n}\sum_{i\in [n]}\varepsilon_{1i}\right)^2,\\
\end{aligned}
$$
where $\dmax$ is the maximal in-degree of network $G$.

By Assumption \ref{asu: bounded potential outcomes} (ii), we have $|\frac{1}{|S_1^{\bw}|}\sum_{i\in S_1^{\bw}}\varepsilon_{1i}|\leq \frac{1}{\sqrt{|S_1^{\bw}|}}c_2$ and $|\frac{1}{n}\sum_{i\in [n]}\varepsilon_{1i}|\leq \frac{1}{\sqrt{n}}c_2$. So 
$$\left|\frac{1}{|S_1^{\bw}|}\sum_{i\in S_1^{\bw}}\varepsilon_{1i} - \frac{1}{n}\sum_{i\in [n]}\varepsilon_{1i}\right|\leq \frac{2}{\sqrt{|S_1^{\bw}|}}c_2,$$
then
$$
I_1\leq 2\E\left(\sum_{d=0}^{d_\text{max}}|f_1(d)||P_{S_1^{\bw}}(d)-P_G(d)|\right)^2 + 2\E \left(\frac{2c_2}{\sqrt{|S_1^{\bw}|}}\right)^2,
$$
by Hölder's inequality,
$$
\begin{aligned}
    I_1
    & \leq 2\E\left(\sum_{d=0}^{d_\text{max}}f_1^2(d)\right)\left(\sum_{d=0}^{d_\text{max}}\left(P_{S_1^{\bw}}(d)-P_G(d)\right)^2\right) + 8c_2^2\E\frac{1}{|S_1^{\bw}|}\\
    & = 2\left(\sum_{d=0}^{d_\text{max}}f_1^2(d)\right) \E\left(\sum_{d=0}^{d_\text{max}}\left(P_{S_1^{\bw}}(d)-P_G(d)\right)^2\right) + 8c_2^2\E\frac{1}{|S_1^{\bw}|}\\
    &= 2||f_1||_2^2 \E||P_{S_1^{\bw}} - P_G||_2^2 + 8c_2^2\E{|S_1^{\bw}|}^{-1}.\\
\end{aligned}
$$
Finally, by Assumption \ref{asu: bounded potential outcomes} (i), we get
$$
\begin{aligned}
    I_1
    & \leq 2 c_1^2 \frac{(d_\text{max}+1)(d_\text{max} + 2)}{2}\E||P_{S_1^{\bw}} - P_G||_2^2 + 8c_2^2\E{|S_1^{\bw}|}^{-1}\\
    & \leq c_1^2(d_\text{max} + 2)^2\E||P_{S_1^{\bw}} - P_G||_2^2 + 8c_2^2\E{|S_1^{\bw}|}^{-1}. 
\end{aligned}
$$
Analogously,
$$
I_0\leq c_1^2(d_\text{max} + 2)^2\E||P_{S_0^{\bw}} - P_G||_2^2 + 8c_2^2\E{|S_0^{\bw}|}^{-1}.
$$
Thus,
$$
    \begin{aligned}
        \MSE(\hat{\tau}) 
        &\leq 2I_1 + 2I_0 \\
        &\leq 4||f_1||_2^2 \E||P_{S_1^{\bw}} - P_G||_2^2 + 4||f_0||_2^2 \E||P_{S_0^{\bw}} - P_G||_2^2 + 16c_2^2\E{|S_1^{\bw}|}^{-1} + 16c_2^2\E{|S_0^{\bw}|}^{-1}\\
        & \leq 2c_1^2\left(d_\text{max}+2\right)^2\left(\E||P_{S_1^{\bw}} - P_G||_2^2 + \E||P_{S_0^{\bw}} - P_G||_2^2\right) + 16c_2^2\left(\E{|S_1^{\bw}|}^{-1}+\E{|S_0^{\bw}|}^{-1}\right).
    \end{aligned}
$$
\end{proof}

\section{Additional results}
\label{app: additional results}

\subsection{Classical resutls of difference-in-means estimators}
\begin{theorem}
    \label{theo: classic resutls of diff}
    Suppose Assumption \ref{asu: fni} holds, then under the RI with CR,
    \begin{itemize}
        \item[(i)] Given $S$, $\hat{\tau}$ is unbiased for $\tau_S$,
        $$\E(\hat{\tau}|S)=\tau_S.$$
        \item[(ii)] Given $S$, $\hat{\tau}$ has variance 
        $$\Var(\hat{\tau}|S)=\frac{V_S^2(\bone)}{|S_1|}+\frac{V_S^2(\bzero)}{|S_0|}-\frac{V_S^2(\tau_S)}{|S|},$$
        where $V_S^2(\bone) = (|S_1|-1)^{-1}\sum_{i\in S_1}(Y_i(\bone)-\bar{Y}_{S_1}(\bone))^2$, $\bar{Y}_{S_1}(\bone)=|S_1|^{-1}\sum_{i\in S_1}Y_i(\bone)$. $V_S^2(\bzero)$ and $\bar{Y}_{S_0}(\bzero)$ are defined analogously, just substituting $\bone$ with $\bzero$. $V_S^2(\tau_S)=(|S|-1)^{-1}\sum_{i\in S}(\tau_i-\tau_S)^2$.
        \item[(iii)] The Mean Squared Error of $\hat{\tau}$ about $\tau$ is \begin{equation}\label{eq: mse decomposition}
            \MSE(\hat{\tau}) = \MSE(\tau_S) + \E_S(\Var(\hat{\tau}|S)),
        \end{equation} 
        where $\MSE(\tau_S) = \E_S(\tau_S-\tau)^2 =  (\E_S(\tau_S)-\tau)^2 + \Var_S(\tau_S)$, $\E_S$ and $\Var_S$ denote expectation and variance taken over all possible $S$.
    \end{itemize}
\end{theorem}
\begin{proof}
    (i) and (ii) follow from the Theorem 4.1 in \cite{ding2024first}. (iii) follows from the tower property of expectations:
    $$\begin{aligned}
        \MSE(\hat{\tau}) 
        &= (\E\hat{\tau}-\tau)^2 + \Var(\hat{\tau}) \\
        &= (\E_S\E(\hat{\tau}|S)-\tau)^2 + \Var_S(\E(\htau|S)) + \E_S(\Var(\hat{\tau}|S)) \\
        &= \{(\E_S(\tau_S)-\tau)^2 + \Var_S(\tau_S)\} + \E_S(\Var(\hat{\tau}|S))\\
        &= \MSE(\tau_S) + \E_S(\Var(\hat{\tau}|S)).
    \end{aligned}
    $$
\end{proof}

\subsection{Matched-pairs randomization (MPR)}
\label{app: matched}
After getting random isolated set $S$, we can also consider matched-pairs randomization. There are many covariates induced from network that can be used to match units. We based it on the simplest one: degrees. We arrange the units in descending order according to their degrees. If $|S|$ is even, we match adjacent units into a pair so there will be $|S|/2$ pairs; otherwise, we let the last ``pair'' include three units (more precisely, this is called stratified randomization) resulting in $\lfloor|S|/2\rfloor$ pairs. We then randomly treat one unit in each pair and proceed with Line \ref{line: cr2} and Line \ref{line: cr3} in Algorithm \ref{alg: cr}. Finally, $|S_1|=\lfloor|S|/2\rfloor$. 

Let $n_{k}$ represent the size of each pair (all equal to 2 if $|S|$ is even, with the last one equal to 3 if not) and let $\htau_k$ denote the standard inverse probability weighted (IPW) estimator for every pair's TTE, where $k=1,\dots,|S_1|$. The matched estimator is
\begin{equation}
    \label{eq: matched estimator}
    \hat{\tau}_m = \sum_{k=1}^{|S_1|}\frac{n_{k}}{|S|}\htau_{k}.
\end{equation}

\subsection{Properties of simple random sample}
\label{app: simple random sample}
\begin{theorem}
\label{theo: simple random sample}
Suppose Assumption \ref{asu: normal bounded outcomes} holds. $\{\tau_i, i\in S\}$ is a simple random sample of $\{\tau_i, i\in [n]\}$ with size $K$. Then
\begin{itemize}
    \item[(i)] $\tau_S$ is unbiased for $\tau$: $$\E(\tau_S)=\tau.$$
    \item[(ii)] $\tau_S$ has variance $$\Var(\tau_S)\leq O(\frac{1}{K}\frac{n-K}{n}).$$
    \item[(iii)] The second term in Equation \ref{eq: mse decomposition} can be controlled: $$\E_S(\Var(\htau|S))\leq O(\frac{1}{K}).$$
    \item[(iv)] Thus the Mean Squared Error of $\htau$: $$\MSE(\htau) \leq O(\frac{1}{K}).$$
\end{itemize}
\end{theorem}
\begin{proof}
    The results follow from Lemma A3.1 in \cite{ding2024first}.
\end{proof}

\end{document}